\documentclass[11pt]{article}
\usepackage{times}
\usepackage{verbatim,amsmath,amssymb,amsthm,enumerate,amsfonts}
\usepackage{graphicx}

\renewcommand{\qed}{\hfill{\rule{2mm}{2mm}}}
\renewenvironment{proof}[1][]{\begin{trivlist}
\item[\hspace{\labelsep}{\bf\noindent Proof#1:\/}] }{\qed\end{trivlist}}

\newcommand{\C}{\mathbb{C}}
\newcommand{\bS}{\overline{S}}
\newcommand{\R}{\mathbb{R}}

\newcommand{\E}{\mathbf{{E}}}

\newcommand{\cE}{\mathcal{{E}}}
\newcommand{\cM}{\mathcal{{M}}}
\newcommand{\cN}{\mathcal{{N}}}

\newcommand{\set}[1]{{\left\{#1\right\}}}
\newcommand{\B}{\set{0,1}}

\newcommand{\Span}{\mathop{\rm Span}\nolimits}

\newcommand{\eps}{\varepsilon}

\newcommand {\xor} {\oplus}

\newcommand{\Tribes}{\mathsf{Tribes}}

\newcommand{\ket}[1]{\left|#1\right\rangle}

\newcommand{\tensor}{\otimes}

\newcommand{\remove}[1]{}
\newcommand{\eqdef}{\stackrel{\rm def}{=}}

\newcommand{\half}{{1 \over 2}}

\newcommand {\wh} {\widehat}

\newcommand{\singeltons}[0]{{\cE_{\mathbf{singeltons}}}}
\newcommand{\bitflip}[0]{{\cE_{\mathbf{cbit}}}}
\newcommand{\phaseflip}[0]{{\cE_{\mathbf{cphase}}}}
\newcommand{\synd}{{\sf synd}}

\theoremstyle{plain}
\newtheorem{theorem}{Theorem}[section]
\newtheorem{lemma}[theorem]{Lemma}

\newtheorem*{claim}{Claim}
\newtheorem{fact}[theorem]{Fact}

\theoremstyle{definition}
\newtheorem{definition}{Definition}[section]

\theoremstyle{remark}

\begin{document}

\title{Approximate quantum error correction for correlated noise}

\author{
Avraham Ben-Aroya\thanks{Department of Computer Science, Tel-Aviv
University, Tel-Aviv 69978, Israel. Supported by the Adams
Fellowship Program of the Israel Academy of Sciences and
Humanities, by the European Commission under the Integrated
Project QAP funded by the IST directorate as Contract Number
015848 and by USA Israel BSF grant 2004390. Email: abrhambe@post.tau.ac.il.}
\and Amnon Ta-Shma\thanks{Department of Computer Science, Tel-Aviv
University, Tel-Aviv 69978, Israel. Supported by the European Commission under the Integrated
Project QAP funded by the IST directorate as Contract Number
015848, by Israel Science Foundation grant 217/05 and by USA Israel BSF grant 2004390.
Email: amnon@tau.ac.il.
}}

\date{}

\maketitle

\null\vspace*{-35pt}

\vspace*{-7pt}

\begin{abstract}
Most of the research done on quantum error correction studies an
error model in which each qubit is affected by noise,
independently of the other qubits. In this paper we study a
different noise model -- one in which the noise may be correlated
with the qubits it acts upon.

We show both positive and negative results. On the one hand, we
show controlled-X errors cannot be \emph{perfectly} corrected, yet
can be \emph{approximately} corrected with sub-constant
approximation error. On the other hand, we show that no
non-trivial quantum error correcting code can approximately
correct controlled phase error with sub-constant approximation
error.
\end{abstract}

\section{Introduction}
One of the reasons for studying quantum error-correcting codes
(QECCs) is that they serve as building blocks for fault-tolerant
computation, and so might serve one day as central components in
an actual implementation of a quantum computer. Much work was done
trying to determine the threshold error, beyond which independent
noise\footnote{The independent noise model is a model in which
each qubit is affected by noise, with some probability,
independently of the other qubits.} can be dealt with by
fault-tolerant mechanisms (see the Ph.D. theses~\cite{R06,A07} and
references therein).

A few years ago there was some debate whether the independent
noise model is indeed a realistic noise model for quantum
computation or not (see, e.g.,~\cite{ALZ06}). This question should
probably be answered by physicists, and the answer to that is most
likely dependent on the actual realization chosen. Yet, while the
physicists try to build actual machines, and the theorists try to
deal with higher independent noise, it also makes sense to try and
extend the qualitative types of errors that can be dealt with. The
results in this paper are both optimistic and pessimistic. On the
one hand, we show there are noise models that can be approximately
corrected but not perfectly corrected, but on the other hand,
there is simple correlated noise that cannot even be approximately
corrected. It might be interesting to reach a better understanding
of what can be approximately corrected. Also, it might be
interesting to come up with other relaxations of quantum error
correction that deal better with correlated noise.

\subsection{Stochastic vs. Adversarial noise}

The basic problem we deal with is that of encoding a message such
that it can be recovered after being transmitted over a noisy
channel. Classically, there are two natural error models: Shannon's
independent noise model and Hamming's adversarial noise model. For
example, a typical noise model that is dealt with in Shannon's
theory, is one where each bit of the transmitted message is flipped
with independent probability $p$, whereas a typical noise model in
Hamming's theory is one where the adversary looks at the transmitted
message and chooses at most $t$ bits to flip. We stress that the
classical \emph{adversarial} noise model allows the adversary to
decide which noise operator to apply based on the specific codeword
it acts upon.

Remarkably, there are classical error correcting codes that solve
the problem in the adversarial noise model, which are almost as
powerful as the best error correcting codes that solve the problem
in the independent noise model. For instance, roughly speaking,
any code in the independent noise model must satisfy $r \le
1-H(p)$, where $r$ is the rate of the code, $p$ is the noise rate.
In the adversarial noise model, the Gilbert-Varshamov bound shows
there are codes with rate $r=1-H(\delta)$ and relative distance
$\delta$, though one can uniquely correct only up to half the
distance.\footnote{If we allow list-decoding, then almost up to
$1-r$ noise rate can be corrected.}

\subsection{The quantum case}

Let us now consider quantum error correcting codes (QECCs). The
standard definition of such codes limits the noise to a linear
combination of operators, each acting on at most $t$ qubits. A
standard argument then shows that a noise operator that acts on $n$
qubits, such that it acts independently on each qubit with
probability $p=t/n$, is very close to a linear combination of error
operators that act on only, roughly, $t$ qubits. Thus, any quantum
error correcting code (QECC) that corrects all errors on at most $t$
qubits, also \emph{approximately} corrects \emph{independent} noise
with noise rate about $t/n$. Therefore, the standard definition of
QECCs works well with independent noise.

As we said before, the classical \emph{adversarial} noise model
allows the adversary to decide which noise operator to apply based
on the specific codeword it acts upon. In the quantum model this
amounts to, say, applying a single bit-flip operator based on the
specific basis element we are given, or, in quantum computing
terminology, applying a controlled bit-flip. Controlled bit-flips
are limited (in that they apply only a single $X$ operator) highly
correlated (in that they depend on all the qubits of the input)
operators. Can QECC correct controlled bit-flip errors? Can QECC
\emph{approximately} correct such errors?

\subsection{Correcting controlled bit flip errors}

Before we proceed, let us first see that a QECC that corrects one
qubit error in the standard sense may fail for controlled bit flip
errors. Assume we have a quantum code of dimension $|C|$ that is
spanned by $|C|$ orthogonal codewords $\set{\phi_i}$, and can
correct any noise operator from $\mathcal{E}$ that is applied on any
vector $\phi \in \set{\phi_i}$. Specifically, for any noise operator
$E\in\mathcal{E}$, the quantum decoding algorithm maps a noisy word
$\phi_i'=E\phi_i$ to a product state $\phi_i \tensor
\ket{\synd(E)}$, where $\synd(E)$ is the error-syndrome associated
with $E$. Notice that $\synd(E)$ depends on $E$ alone and not on
$\phi_i$. Then, we can correct errors applied on any state in the
vector space \emph{spanned} by the basis vectors $\set{\phi_i}$. To
see that, notice that if we start with some linear combination
$\sum_{i=1}^k \alpha_i {\phi_i}$ and we apply the error $E$ on it,
then the corrupted word is $\sum_{i=1}^k \alpha_i E\phi_i$, and
applying the decoding procedure we get the state $(\sum_{i=1}^k
\alpha_i \phi_i) \tensor \ket{\synd(E)}$. Tracing out the syndrome
register we recover the original state. This property, however,
breaks down for controlled bit flip errors, where the error may
depend on the specific codeword $\phi_i$. In that case the corrupted
word is $\sum_{i=1}^k \alpha_i E_i\phi_i$. If we use the same
decoding procedure, and the decoded word is $\sum_{i=1}^k \alpha_i
\phi_i \tensor \ket{\synd(E_i)}$ and if we trace out the syndrome
register we end up with a state different then the original state.

The above argument shows that if we allow controlled bit-flip
errors, then the environment may get information about the codeword,
and thus corrupt it. This, by itself, is not yet an impossibility
proof, as it is possible that one can find a code that is immune to
controlled bit-flip errors. Unfortunately, an easy argument shows
that there is no non-trivial QECC that perfectly corrects such
errors (see Theorem~\ref{thm:exact-negative}). Therefore, while
there are asymptotically good QECC in the standard error model,
there are no non-trivial QECC correcting controlled \emph{single}
bit-flip errors.

%In short, a classical code may handle adversarial noise because it it is a
%discrete set of some $|C|$ words, whereas a quantum code cannot handle
%adversarial noise because it is a \emph{vector space} of dimension $|C|$ of
%valid codewords.

\subsection{Approximate error-correction}
Summarizing the discussion above, we saw that no QECC can
\emph{perfectly} correct controlled bit-flip errors. We now ask
whether this also holds when we relax the perfect decoding
requirement and only require \emph{approximate decoding}. Namely,
suppose we only require that for any codeword $\phi$ and any allowed
error $E$, decoding $E\phi$ results in a state \emph{close} to
$\phi$. Can we then correct controlled bit-flip errors?

Somewhat surprisingly we show a positive answer to this question.
That is, we show a QECC of arbitrarily high dimension, that can
correct any controlled bit-flip error with sub-constant
approximation-error (see Theorem~\ref{thm:app-positive} for a
formal statement). This, in particular, shows that there are error
models that cannot be perfectly decoded, yet can be approximately
decoded. For the proof, we find a large dimension vector space
containing only low-sensitive functions.

Having that we increase our expectations and ask whether one can
approximately correct, say, any controlled single-qubit error.
However, here we show a negative result. We show that no
non-trivial QECC can correct controlled phase-errors with only
sub-constant approximation error (see
Theorem~\ref{thm:app-negative} for a formal statement). Namely, no
non-trivial QECC can handle, even approximately, correlated noise,
if the control is in the standard basis and the error is in the
phase.

%However, we mention that it is plausible that the errors are basis dependent. For example, assume each gate is associated with some typical error. One way to think of a CNOT gate is that it is the identity when its input is $\ket{0}$ and $X$ when its $\ket{1}$. Thus, it is natural to expect a dependence on the input, with regard to the eigenvector basis of the CNOT operator. In general, if we apply at time $T$ an operator $T_1 \tensor T_2 \tensor T_k$ it is natural to expect basis dependence with regard to the eigenvector basis of $T$. Indeed, if this is the case quantum error correction should work well when using only local gates. The big question, which should be answered by physicists, is whether one can indeed design a system

\section{Preliminaries}
\label{sec:preliminiaries}
\subsection{Quantum error-correcting
codes} Let $\cN$ denote the Hilbert space of dimension $2^n$. $\cM$
is a $[n,k]$ quantum error correcting code (QECC) if it is a
subspace of $\cN$ of dimension $K \ge 2^k$. We call $n$ the
\emph{length} of code, and $K$ the \emph{dimension} of the code. For
two Hilbert spaces $\cN,\cN'$, $L(\cN,\cN')$ denotes the set of
linear operators from $\cN$ to $\cN'$.

\begin{definition}
A code $\cM$ \emph{corrects} $\cE \subset L(\cN,\cN')$ if for any
two operators $X,Y \in \cE$ and any two codewords $\phi_1,\phi_2 \in
\cM$, if $\phi_1^* \phi_2 = 0$ then  $(X\phi_1)^*(Y\phi_2) = 0$.
\end{definition}
\begin{fact}[{\cite[Section 15.5]{KSV02}}]\label{fact:QECC}
A code $\cM$ corrects $\cE$ if for any $X,Y \in \cE$, defining
$E=X^*Y$, there exists a constant $c(E) \in \C$, such that for any
two codewords $\phi_1,\phi_2 \in \cM,$ $$\phi_1^*E\phi_2=c(E)
\cdot \phi_1^* \phi_2.$$
\end{fact}

A QECC $\cM$ corrects $t$ errors if it corrects all linear operators
that correlate the environment with at most $t$ qubits. There are
\emph{asymptotically good} QECCs, i.e., $[n,k]$ QECCs that correct
$t=\Omega(n)$ errors with $n=O(k)$~\cite{ALT01}.

\subsection{Boolean functions}
The influence of a variable $x_i$ on a boolean function $f:\B^n \to
\B$ is defined to be $$\Pr_{x \in \B^n} [f(x) \neq f(x \xor e_i)],$$
where $e_i \in \R^n$ is the $i$'th vector in the standard basis. The
influence of a function is the maximum influence of its variables.
Ben-Or and Linial~\cite{BL85} showed that there exists a balanced
function $\Tribes:\B^n \to \B$ with influence as small as $O({\log n
\over n})$, and Kahn, Kalai and Linial~\cite{KKL88} showed that this
bound is tight for balanced functions. We extend this notion to
complex valued functions. For $g:\B^n \to \C$ let
$$I_i(g)=\E_{x \in \B^n} |g(x)-g(x \xor e_i)|^2$$ and $I(g) =
\max_{i \in [n]} I_i(g)$.

We identify a function $g:\B^n \to \C$ with the vector $\sum_{x
\in \B^n} g(x) \ket{x}$. When we write $g$ we refer to it as a
vector in $\cN$. When we write $g(x)$ we refer to $g$ as a
function $g:\B^n \to \C$ and $g(x) \in \C$.

\section{No QECC can perfectly correct controlled bit flips}
\label{sec:no-perfect-QECC}

We now concentrate on the error model that allows any
\emph{controlled bit flip} error. Formally, for $i \in [n]$ and
$S\subseteq \B^{n-1}$ let $E_{i,S}$ be the operator that applies
$X$ on the $i$'th qubit conditioned on the other qubits being in
$S$. More precisely, we define the operator $E_{i,S}$ on the basis
$\set{\ket{x}| x \in \B^n}$ and extend it linearly. For $x \in
\B^n$ define $\wh{x}_i=(x_1, \ldots, x_{i-1},x_{i+1},\ldots,x_n)
\in \B^{n-1}$. Also, let $X^i \in L(\cN,\cN)$ denote the operator
that flips the $i$'th qubit, i.e., $X^i = I^{\tensor (i-1)} \tensor X
\tensor I^{\tensor (n-i)}$. Then
$$E_{i,S} \ket{x} = \left\{%
\begin{array}{ll}
    X^i \ket{x} & ~~~\hbox{if $\wh{x}_i \in S$} \\
    \ket{x} & ~~~\hbox{otherwise.} \\
\end{array}%
\right.$$ Let
$$\bitflip \eqdef \set{E_{i,S} ~~|~~ i \in [n],~ S \subseteq \B^{n-1} }.$$
We also define a tiny subset $\singeltons$ of $\bitflip$ by
$$\singeltons \eqdef \set{E_{i,\set{j}} ~~|~~ i \in [n],~ j \in \B^{n-1} }.$$
We claim that even this set of errors \emph{cannot} be corrected.

\begin{theorem}
\label{thm:exact-negative} There is no QECC with dimension bigger
than one that can correct $\singeltons$.
\end{theorem}

\begin{proof}
Suppose there exists a $[n,k]$ code with $k \ge 1$ that corrects
$\singeltons$. Let $\phi = \sum_{i \in \B^n} \phi(i) \ket{i}$ and
$\psi = \sum_{i \in \B^n} \psi(i) \ket{i}$ be two orthonormal
codewords. We will prove that:

\begin{claim}
For every $i \in [n]$ and every $q \in \B^n$ it holds that $\phi(q)=\phi(q \xor e_i)$.
\end{claim}
In particular, it follows that $\phi=\alpha \cdot \sum_{i \in
\B^n} \ket{i}$ for some $0 \neq \alpha \in \C$. Similarly,
$\psi=\alpha' \cdot \sum_{i \in \B^n} \ket{i}$ for some $0 \neq
\alpha' \in \C$. Therefore, $\phi^*\psi  = 2^n \alpha^* \alpha'
\neq 0$. A contradiction.

We now prove the claim. Fix $i \in [n]$ and $q \in \B^{n}$. Denote
$E=E_{i,\set{q}}$ and $q'=q \xor e_i$. It can be verified that
\begin{eqnarray*}
\phi^* E \phi &=& \phi^* \phi - |\phi(q) - \phi(q')|^2 =1-|\phi(q) - \phi(q')|^2 \\
\psi^* E \psi &=& \psi^* \psi - |\psi(q) - \psi(q')|^2 =1-|\psi(q) - \psi(q')|^2 \\
\psi^* E \phi &=&
%\psi^* \phi - (\phi(q) - \phi(q'))^*(\psi(q) - \psi(q')) =
- (\phi(q) - \phi(q'))^*(\psi(q) - \psi(q'))
\end{eqnarray*}

As $\phi^* \psi =0$, by the QECC definition, $\psi^* E\phi = 0$
and so $(\phi(q) - \phi(q'))^*(\psi(q) - \psi(q'))=0$. If $\phi(q)
\neq \phi(q')$ we conclude that $\psi(q)=\psi(q')$. But then,
$\phi^* E\phi < 1$ while $\psi^* E \psi=1$, which contradicts
Fact~\ref{fact:QECC}. Therefore, $\phi(q)=\phi(q')$.
\end{proof}

The argument above shows that for any $w,w' \in \B^n$ and any
codeword $\phi$, $\phi(w)=\phi(w')$, by employing a sequence of
\emph{small} changes, and showing that $\phi$ is invariant under
these small changes. However, if we replace the stringent notion perfect decoding with the more relaxed
notion of \emph{approximate decoding}, then at least theoretically it is possible that under this weaker
notion, controlled bit flips can be corrected. Somewhat
surprisingly, this is indeed the case.

\section{An approximate QECC for controlled bit flips}

We first define a relaxed notion of error-detection. We say a code \emph{separates} $\cE$ if
for any two allowed errors $X,Y \in \cE$ and any two orthogonal
codewords $\phi,\psi$, $X\phi$ and $Y\psi$ are far away from each
other. Formally,
\begin{definition}
Let $\cM \subseteq \cN$ be an $[n,k]$ QECC and $\cE \subset
L(\cN,\cN')$. We say $\cM$ separates $\cE$ with at most $\alpha$
error, if for any two operators $X,Y \in \cE$ and any two unit
vectors $\phi_1, \phi_2 \in \cM$, if $\phi_1^* \phi_2 = 0$ then
$|\phi_1^*X^* Y\phi_2| \le \alpha$.
\end{definition}

We say a code $\cM$ \emph{approximately} corrects $\cE$ if there exists a
POVM on $\cN'$  such that for any operator $X \in \cE$, and any
codeword $\phi \in \cM$, when we apply the POVM on $X\phi$, the
resulting mixed state is close to the pure
state $\phi$. A very special case of the above is when the decoding
procedure is the \emph{identity} function. In this case we say $\cM$
is $(\cE,\eps$) \emph{immune}. Formally,

\begin{definition}
Let $\cM \subseteq \cN$ be an $[n,k]$ QECC and $\cE \subset
L(\cN,\cN')$. We say $\cM$ is $(\cE,\eps$) immune if for every $X
\in \cE$ and every $\phi \in \cM$, $|\phi^*X\phi| \ge
(1-\eps)|\phi^* \phi|$. We call $\eps$ the \emph{approximation
error}.
\end{definition}

We saw before that there is no non-trivial QECC that perfectly corrects $\bitflip$. In contrast, we will now construct a large QECC that is immune against $\bitflip$, with sub-constant approximation error.

\subsection{The construction}

The calculations done in Section~\ref{sec:no-perfect-QECC} can be
generalized to show that if we want $\phi$ to be $\eps$-immune for
bit flip errors, then $\phi(x)$ must have low influence. However,
we also want the QECC to have a large dimension, and so we want many
orthogonal such vectors. The idea is to work with a function $f$ of
low influence, and combine it on many independent blocks.

Pick an integer $B$ such that $2B$ divides $n$, and define $n'={n
\over 2B}$. Fix a balanced function $f:\B^{n'} \to \set{\pm \half}$
with low influence, i.e., $I(f) \le s=s(n')$. We remind the reader
that this means that for all $j$ in $[n']$, $I_j(f)=\E_{x \in
\B^{n'}} |f(x)-f(x \xor e_j)|^2 \le s$ (see
Section~\ref{sec:preliminiaries}). Notice that this implies that for
all $j \in [n']$,
\begin{eqnarray}
\label{eqn:pr}
\Pr_{w \in \B^{n'}} [f(w) \neq f(w \xor e_j)] & \le & s.
\end{eqnarray}
We use the low-influence function $f$ as a building block.

Now partition $[n]$ into $2B$ blocks of equal length $n'$. For $x
\in \B^n$, $i \in \set{1,\ldots,B}$ and $b \in \B$, let $x_{i,b} \in
\B^{n'}$ denote the value of $x$ restricted to the $(2i-1+b)$'th
block, i.e., the string  $x$ is the concatenation of the blocks
$x_{1,0},x_{1,1},x_{2,0},x_{2,1},x_{3,0},\ldots,x_{B,0},x_{B,1}$.
For $z=(z_1,\ldots,z_B) \in \B^B$ we define a function $f_z: \B^n \to \C$
that apply $f$ on the blocks corresponding to $z$. That is,
$$f_z(x)=f(x_{1,z_1}) \cdot \ldots \cdot f(x_{B,z_B}),$$
as shown in Figure~\ref{fig:ctrlX-QECC}.

\begin{figure}[ht]
    \centering
    \begin{minipage}[c]{.45\textwidth}
        \centering
       \includegraphics[scale=0.5]{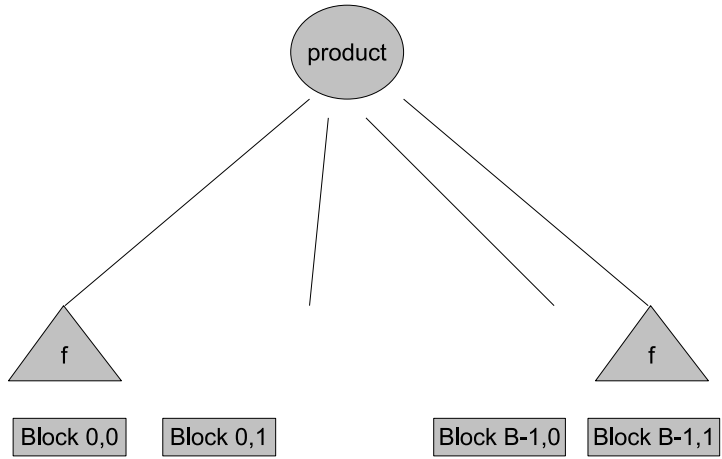}
    \end{minipage}
    \hfill
    \begin{minipage}[c]{.3\textwidth}
        \centering
        \caption{\it The input is $B$ pairs of blocks, each block is of length $n'$. The values $z_1,\ldots,z_{B} \in \B^B$ determine on which block in each pair $f$ is applied. In the example, $z_1=0$ while $z_{B}=1$. The output is the product of the $B$ values.}
         \label{fig:ctrlX-QECC}
    \end{minipage}
\end{figure}

As usual we look at $f_z$ as a vector in $\cN$. We let $W=\Span
\set{f_z ~:~ z \in \B^B}$. We claim:

\begin{theorem}
\label{thm:1-bit-flip} $W$ is an $[n,B]$ QECC that is
$(\bitflip,{2s(n')})$ immune.
\end{theorem}

In particular, taking $f(w)=\half$ when $\Tribes(w)=1$ and
$f(w)=-\half$ when $\Tribes(w)=0$, we get:

\begin{theorem}
\label{thm:app-positive}
For every $0 < B=B(n) < n$ there exists
an $[n,B]$ QECC that is $(\bitflip,\eps=O({B \log n \over n}))$ immune.
\end{theorem}

In particular, there exists QECCs of length $n$ and dimension
$2^{\sqrt{n}}$ that approximately corrects all controlled-X errors with
an $O({\log(n) \over \sqrt{n}})$ approximation error.

\subsection{The analysis}

We first show that $\dim W = B$. This immediately follows from:

\begin{claim}
$\set{f_z}_{z \in \B^B}$ is an orthogonal set.
\end{claim}

\begin{proof}
%Clearly, $f_z^*f_z=2^{n-2B}$
We will
show that for $z \neq z'$, $f_z^* f_{z'}=0$. For that, it is enough
to show that $\set{(f_z(x),f_{z'}(x))}_{x \in \B^n}$ is uniform over
$\set{\pm 2^{-B}} \times \set{\pm 2^{-B}}$.

To see that, first notice that $f(x_{1,i_1})$ is balanced over
$\set{\pm \half}$. Hence, $\set{f_z(x)}_{x \in \B^n}$ is uniform
over $\set{\pm 2^{-B}}$. Also, as $z \neq z'$ there exists some $k$
such that $z_k \neq z'_k$. Notice that $f(x_{k,z_k})$ depends on
bits that do not influence $f_{z'}(x)$, hence it is independent of
$f_{z'}(x)$. It is also uniform on $\set{ \pm \half}$. Hence the
pair $(f_{z}(x),f_{z'}(x))$ is uniform over $(\pm 2^{-B},\pm
2^{-B})$ as desired.
\end{proof}

We now analyze the approximation error. We will use the following
lemmas:
\begin{lemma}
\label{lem:sensitivity} For every $\phi \in \cN$ and every $i \in
[n]$, $S \subseteq \B^{n-1}$,
$$|\phi^* E_{i,S} \phi -\phi^*\phi| \le 2^{n-1} I_{i}(\phi).$$
\end{lemma}

\begin{lemma}
\label{lem:W} For every $\phi \in W$, $$2^{n-1} I(\phi) \le {2s}
\cdot |\phi^*\phi|.$$
\end{lemma}

These lemmas together imply Theorem~\ref{thm:1-bit-flip}. Notice
that in Lemma~\ref{lem:W} we had to prove the claim for every $\phi
\in W$ and not just for some basis of $W$ (see the discussion in the
introduction).

\begin{proof}[ of Lemma~\ref{lem:sensitivity}]
For every $g,h:\B^n \to \C$ and every $i \in [n]$, $S \subseteq \B^{n-1}$,
\begin{eqnarray*}
{h}^* E_{i,S} {g}  & = &
\sum_{x:\wh{x}_i \not \in S} h(x)^*g(x)+
\sum_{x:\wh{x}_i  \in S} h(x)^*g(x \xor e_i) \\
& = & \sum_{x \in \B^n} h(x)^*g(x)+\sum_{x: \wh{x}_i \in S}
[h(x)^*g(x \xor e_i)-h(x)^*g(x)].
\end{eqnarray*}

Now, fix $i$. For $y \in \B^{n-1}$ and $b \in \B$ let $(y,b)$ denote
the string $x \in \B^n$ such that $\wh{x}_i=y$ and $x_i=b$. Then,
\begin{eqnarray*}
|{h}^* E_{i,S} {g}  -h^*g|
& = & \left|\sum_{y \in S} (h(y,0)^*-h(y,1)^*)(g(y,0)-g(y,1))\right| \\
& \le &
\sqrt{\sum_{y \in S}|h(y,0)^*-h(y,1)^*|^2} \sqrt{\sum_{y \in S}|g(y,0)-g(y,1)|^2} \\
%%%%%%%%%%%%%%%%%%%%%%%%%%%%%%%%%%%%%%%%
%& = &
%\sqrt{4|S| 2^{-n} \Pr_{y \in S} (\alpha(y)=\alpha(y %\xor %e_i))}
%\sqrt{4|S| 2^{-n} \Pr_{y \in S} (\beta(y)=\beta(y \xor e_i))} \\
%& = & 2 \rho(S) \sqrt{I_{S,i}(\alpha)} \sqrt{I_{S,i}(\beta)}
%%%%%%%%%%%%%%%%%%%%%%%%%%%%%%%%%%%%%%%%
& \le &
\sqrt{\sum_{y \in \B^{n-1}}|h(y,0)^*-h(y,1)^*|^2} \sqrt{\sum_{y \in \B^{n-1}}|g(y,0)-g(y,1)|^2} \\
& = & \sqrt{2^{n-1} I_i(h)}\sqrt{2^{n-1} I_i(g)}
\end{eqnarray*}
\end{proof}

We now turn to Lemma~\ref{lem:W}. One can check that all elements in
$\set{f_z}$ have low influence. However, this by itself does not
imply that all elements in $W$ are so. So we verify this directly.

\begin{proof}[ of Lemma~\ref{lem:W}]
We want to show that any $\phi \in W$ has low influence. Fix $i \in
[n]$ and suppose that $i$ corresponds the $j$'th variable in the
$(k,b)$'th block. For $x \in \B^n$ let $x=(x^{(1)},x^{(2)})$ where
$x^{(2)}=x_{k,b}$ and $x^{(1)} \in \B^{n-n'}$ is the rest of $x$. Let
$\wh{f_z}:\B^{n} \to \C$ be $$\wh{f_z}(x)=f(x_{1,z_1}) \cdot \ldots
f(x_{k-1,z_{k-1}}) \cdot f(x_{k+1,z_{k+1}}) \cdot \ldots
f(x_{B,z_B}).$$ Notice that $\wh{f_z}(x^{(1)},x^{(2)})$ depends only on
$x^{(1)}$. For that reason we also write it as $\wh{f_z}(x^{(1)})$.

We are given $\phi \in W$ and express it as $\phi=\sum_{z}
\alpha_{z} f_{z}$. We want to bound
\begin{eqnarray*}
I_i(\phi) &=& \E_{x \in \B^{n}}|\phi(x)-\phi(x \xor e_i)|^2 \\
& = & \E_{x \in \B^{n}} \left|\sum_z \alpha_z (f_z(x)-f_z(x \xor
e_i))\right|^2.
\end{eqnarray*}

The functions $f_z$ for which $f_z(x)=f_z(x \xor e_i)$ do not
contribute to the sum. We can therefore define $\zeta=\sum_{z: z_k=b}
\alpha_{z} f_{z}$ and it follows that $I_i(\phi)=I_i(\zeta)$. Then,
\begin{eqnarray*}
2^{n-1} I_i(\zeta) &=& \half \sum_{x \in \B^{n}} \left|\sum_{z:z_k=b} \alpha_z (f_z(x)-f_z(x \xor e_i))\right|^2\\
& = & \half \sum_{(x^{(1)},x^{(2)})\in \B^{n}} \left|\sum_{z:z_k=b} \alpha_z \wh{f_z}(x^{(1)})(f(x^{(2)})-f(x^{(2)} \xor e_j))\right|^2\\
& = & \half \sum_{(x^{(1)},x^{(2)})\in \B^{n}} |f(x^{(2)})-f(x^{(2)} \xor e_j)|^2
\cdot \left|\sum_{z:z_k=b} \alpha_z \wh{f_z}(x^{(1)})\right|^2.
\end{eqnarray*}

Next, observe that the only terms that contribute non-zero values
are those $x=(x^{(1)},x^{(2)})$ where $f(x^{(2)}) \neq f(x^{(2)} \xor e_j)$. There
are at most $s2^{n'}$ such strings $x^{(2)}$ (see
Equation~(\ref{eqn:pr})). Also, each such term contributes
$|\sum_{z:z_k=b} \alpha_z \wh{f_z}(x^{(1)})|^2$. However,
\begin{eqnarray*}
|\zeta(x^{(1)},x^{(2)})|^2 &=& \Big|\sum_{z: z_k=b} \alpha_{z} f_{z}(x^{(1)},x^{(2)})\Big|^2
~=~\Big|\sum_{z: z_k=b} \alpha_{z} \wh{f_z}(x^{(1)})f(x^{(2)})\Big|^2 \\
&=& \big|f(x^{(2)}) \big|^2 \cdot \Big|\sum_{z: z_k=b} \alpha_{z} \wh{f_z}(x^{(1)})\Big|^2
~=~ {1 \over 4}\Big|\sum_{z:z_k=b} \alpha_z \wh{f_z}(x^{(1)})\Big|^2.
\end{eqnarray*}

Thus, the term $|\zeta(x^{(1)},x^{(2)})|^2$ depends only on $x^{(1)}$ and not on
$x^{(2)}$, and we denote it by $|\zeta(x^{(1)})|^2$. Denote
$\wh{\zeta}=\sum_{z: z_k=b} \alpha_{z} \wh{f_z}$. Notice that
$\sum_{x^{(1)}} |\zeta(x^{(1)})|^2=|{\wh{\zeta}~ }^* \wh{\zeta}|$. Also, because
$\zeta(x^{(1)},x^{(2)})$ does not depend on $x^{(2)}$, $|\zeta^* \zeta| = 2^{n'}
|{\wh{\zeta} ~}^* \wh{\zeta}|$. Altogether,
\begin{eqnarray*}
2^{n-1} I_i(\zeta) &=&  \half \sum_{x_1} s2^{n'} 4|\zeta(x_1)|^2
={2s}\cdot 2^{n'} |{\wh{\zeta}~}^* \wh{\zeta}| = {2s}|\zeta^*\zeta|.
\end{eqnarray*}

Finally, $\zeta=\sum_{z:z_k=b} \alpha_z f_z$ is a linear combination
of orthogonal functions $\set{f_z}$ and $|f_z^*f_z|=2^{n-2B}$. Thus,
$$|\zeta^*\zeta|=\sum_{z:z_k=b} |\alpha_z|^2 2^{n-2B} \le 2^{n-2B}
\sum_z |\alpha_z|^2=|\phi^*\phi|,$$ which completes the proof.
\end{proof}

\section{No approximate QECC can correct controlled phase errors}

So far we have seen that one can approximately correct controlled-X
errors with a sub-constant approximation error. We now show there is
no way to correct controlled-phase errors. The reason is that if
$\phi$ and $\psi$ are two orthogonal codewords, then by applying
controlled phase errors we can match the phase of $\phi$ and $\psi$
on any basis vector $\ket{x}$, and this implies  that
$|\phi^*X\psi|$ is about $\sum_x |\phi(x)| \cdot |\psi(x)|$ which
leads to a simple contradiction.

We now formally define our error model. As before, the error
operators are linear and hence it is sufficient to define them on
the standard basis $\set{\ket{x}}_{x \in \B^n}$. We define the error
operators $E_{S,\theta}$, for $S \subseteq \B^n$ and $\theta \in
[2\pi]$ by:
$$E_{S,\theta} \ket{x} = \left\{%
\begin{array}{ll}
    e^{\theta i} \ket{x} & ~~~\hbox{if $x \in S$} \\
    \ket{x} & ~~~\hbox{otherwise.} \\
\end{array}%
\right.$$

In fact, we do not even need to allow any controlled phase error,
and we can be satisfied with $\theta \in \set{0,{\pi \over 4},{\pi
\over 2}}$. Set,
$$\phaseflip = \set{E_{S,\theta} ~|~ S \subseteq \B^n, \theta \in  \set{0,{\pi \over 4},{\pi
\over 2}}}.$$

We now prove that such errors cannot be approximately corrected,
even for some fixed constant error. In fact, we prove that such
errors cannot even be \emph{separated}.

\begin{theorem}
\label{thm:app-negative} There is no two-dimensional QECC that
separates $\phaseflip$ with at most $\alpha={1 \over 10}$ error.
\end{theorem}

\begin{proof}
Assume $\cM$ be a 2-dimensional QECC that separates $\phaseflip$
with at most $\alpha$ error.

\begin{lemma}\label{lem:abs}
Let $\cM \subseteq \cN$ be a vector space of dimension greater than
one. Then there are two orthonormal vectors $\phi,\psi \in \cM$ such
that $\sum_x |\phi(x)| \cdot |\psi(x)| \ge \half$.
\end{lemma}

We postpone the proof for later. Fix $\phi$ and $\psi$ as in the
lemma. Notice that by ranging over all $X,Y \in \phaseflip$, we can
implement any operator $E$ that partitions $\B^n$ to four sets, and
based on the set does a phase shift of angle $0$, ${\pi \over 4}$,
${\pi \over 2}$ or ${3\pi \over 4}$. More precisely: for a partition
$\bS = (S_1,\ldots,S_\ell)$ of $\B^n$ and for a tuple of angles
$\Theta=(\theta_1,\ldots,\theta_\ell) \subseteq [2\pi]^\ell$ define
$E_{\bS,\Theta}$ by $E_{\bS,\Theta} \ket{x}= e^{\theta_j i}
\ket{x}$, where $j$ is such that $x \in S_j$. Let $\Theta_k=\left(
0, {1 \over 2^{k}} \pi,\ldots, {(2^{k}-1) \over 2^{k}} \pi\right)$
and
$$\phaseflip_k = \set{E_{\bS,\Theta_k} ~|~ \bS=(S_1,\ldots,S_{2^k}) \mbox{ is a partition of $\B^n$}}.$$

Then, by ranging over all $X,Y \in \phaseflip$ we range over all $E
\in \phaseflip_2$. We claim:

\begin{lemma}\label{lem:phase}
Let $k \ge 1$ and $\eps=2^{-k}\pi$. For every $\phi,\psi \in \cN$
there exists $E \in \phaseflip_k$ such that $$|\phi^* E \psi| \ge
(1-\eps)\sum_x |\phi(x)| \cdot |\psi(x)|.$$
\end{lemma}

Thus, in particular for the $\phi$ and $\psi$ we fixed before (and
setting $k=2, \eps={\pi \over 4}$):
$$\alpha \ge |\phi^* X \psi| \ge {1-\eps \over 2} >
\frac{1}{10}.$$
\end{proof}

We are left to prove the two lemmas:

\begin{proof}[ of Lemma~\ref{lem:abs}]
Let $\phi,\psi \in \cM$ be arbitrary orthonormal vectors. Let
$\phi'={1 \over \sqrt{2}}(\phi+\psi)$ and $\psi'={1 \over
\sqrt{2}}(\phi-\psi)$. Then
\begin{eqnarray*}
\sum_x |\phi'(x)| \cdot |\psi'(x)| &=& \half \sum_x
|\phi(x)+\psi(x)| \cdot |\phi(x)-\psi(x)|.
\end{eqnarray*}

Fix some $x \in \B^n$. Denote $a=\phi(x), b=\psi(x)$, $a,b \in \C$
and assume $|a| \ge |b|$. Then,
\begin{eqnarray*}
|(a+b)(a-b)| &=& |a^2-b^2| \ge |a|^2-|b|^2 = (|a|-|b|) (|a|+|b|) \\
& \ge & (|a|-|b|)^2 = |a|^2+|b|^2-2 |a| \cdot |b|.
\end{eqnarray*}

Therefore,
\begin{eqnarray*}
\sum_x |\phi'(x)| \cdot |\psi'(x)| &\ge& \half \sum_x
\left(|\phi(x)|^2+|\psi(x)|^2\right) -\sum_x |\phi(x)| \cdot
|\psi(x)| = 1-\sum_x |\phi(x)| \cdot |\psi(x)|.
\end{eqnarray*}

Thus, either $\sum_x |\phi(x)| \cdot |\psi(x)|$ or $\sum_x
|\phi'(x)| \cdot |\psi'(x)|$ is at least half.
\end{proof}

\begin{proof}[ of Lemma~\ref{lem:phase}]
Express $\phi(x)=r_x \cdot e^{\theta_x i}$ with $r_x=|\phi(x)| \in
R^{+}$ and $\theta_x \in [2\pi]$. Similarly, $\psi(x)=r'_x \cdot
e^{\theta'_x i}$. The partition $\bS=(S_1,\ldots,S_{2^k})$ is
defined as follows. For every $x$ we look at $\min_j \set{|\theta'_x
+ \theta_j - \theta_x |} $. Any $x$ is chosen to be in $S_{j'}$
according to the $j'$ that minimizes the above expression for that
$x$. Note that the above expression is always bounded by
$2^{-k}\pi$.

Denote $E=E_{\bS,\Theta_k}$. Then, $$|\phi^*E\psi|=  \left|\sum_x
r_x r'_x e^{\zeta_x i}\right|,$$ where $|\zeta_x| \le 2^{-k}\pi$.
Let $u_x=1-e^{\zeta_x i}$ and notice that
$$|u_x|^2=2(1-\cos(\zeta_x)) \le \zeta_x^2 \le 2^{-2k}\pi^2$$ and
$|u_x| \le 2^{-k}\pi=\eps$. Thus,
\begin{eqnarray*}
|\phi^*E\psi| &=& \left|\sum_x r_x r'_x (1-u_x)\right| ~\ge~ \sum_x r_x r'_x -\left|\sum_x r_x r'_x u_x \right| \\
&\ge & (1-\max_x |u_x|)\sum_x r_x r'_x ~\ge~ (1-\eps)\sum_x r_x r'_x,
\end{eqnarray*}
as desired.
%
% I used the Taylor expansion for cos. cos (t) \ge 1-t^2/2. 2(1-cos(t)) \le 2(1-(1-t^2/2)) = t^2
\end{proof}

%We now turn to proving that no non-trivial QECC can have this
%property. The key point we use here is that a QECC is vector space.

\bibliographystyle{alpha}
\bibliography{refs}
\end{document}